\documentclass[format=acmsmall, review=false, screen=true]{acmart}

\usepackage{acm-ec-17}
\usepackage{booktabs} 
\usepackage{tikz, algorithm, bm}

\usepackage[noend]{algpseudocode}
\newcommand{\alg}[2]{\begin{algorithm}[ht] \caption{#1} \label{alg:#2}
		\begin{algorithmic}[1]}
		\newcommand{\ealg}{\end{algorithmic} \end{algorithm}}
\newcommand{\aref}[1]{\textbf{Algorithm}~\ref{alg:#1}}

\renewcommand{\b}[1]{\left[#1\right]}

\theoremstyle{definition}

\newtheorem{claim}{Claim}
\usetikzlibrary{positioning,decorations.pathmorphing,shapes,decorations.pathreplacing}
\newcommand\numberthis{\addtocounter{equation}{1}\tag{\theequation}}
\newcommand{\p}[1]{\left(#1\right)}

\newcommand{\argmin}{\operatornamewithlimits{argmin}}

\newcommand{\eps}{\varepsilon} 
\newcommand{\ra}{\rightarrow} 
\newcommand{\omt}[1]{ } 

\newcommand{\xhdr}[1]{\paragraph*{\bf {#1}}}

\newcommand{\prevproof}[3]{
	{\vskip 0.1in \noindent {\bf Proof of {#1}~\ref{#2}.} {#3} \rule{2mm}{2mm}
		\vskip \belowdisplayskip}
}

\newcommand{\camera}[1]{ } 
\newcommand{\full}[1]{ #1} 

\begin{document}

\title{Planning with Multiple Biases} 
\author{Jon Kleinberg}
	\affiliation{
		\institution{Cornell University}
	}
	\author{Sigal Oren}
	\affiliation{
		\institution{Ben-Gurion University of the Negev}
	}
	\author{Manish Raghavan}
	\affiliation{
	\institution{Cornell University}
   }

\thanks{
	A full version is available from \url{ https://arxiv.org}. Supported in part by a Simons Investigator Award, an ARO MURI grant, a Google Research Grant, and a Facebook Faculty Research Grant.
}
\begin{abstract}
Recent work has considered theoretical models for the behavior of agents with specific behavioral biases: rather than making decisions that optimize a given payoff function, the agent behaves inefficiently because its decisions suffer from an underlying bias. These approaches have generally considered an agent who experiences a single behavioral bias, studying the effect of this bias on the outcome. 

In general, however, decision-making can and will be affected by multiple biases
operating at the same time. How do multiple biases interact to produce the
overall outcome? Here we consider decisions in the presence of a pair of biases
exhibiting an intuitively natural interaction: present bias -- the tendency to
value costs incurred in the present too highly -- and sunk-cost bias -- the
tendency to incorporate costs experienced in the past into one's plans for the
future. 

We propose a theoretical model for planning with this pair of biases, and we
show how certain natural behavioral phenomena can arise in our model only when
agents exhibit both biases. As part of our model we differentiate between agents
that are aware of their biases (sophisticated) and agents that are unaware of
them (naive). Interestingly, we show that the interaction between the two biases
is quite complex: in some cases, they mitigate each other's effects while in other cases they might amplify each other. We obtain a number of further results as well, including the fact that the planning problem in our model for an agent experiencing and aware of both biases is computationally hard in general, though tractable under more relaxed assumptions. 
\end{abstract}

\maketitle

\section{Introduction}
A rich genre of work at the interface of economics and psychology
has studied the ways in which behavioral and cognitive biases 
can lead people to make consistently sub-optimal decisions
\cite{della-vigna-psych-econ,kahneman-thinking, ariely2008predictably,thaler2015misbehaving}.
Research in this area has provided a useful organization of
these types of biases, including broad categories such as 
treating losses and gains asymmetrically
\cite{kahneman-prospect-theory},
treating the present inconsistently relative to the future
\cite{frederick-time-inconsist-surv},
and systematically mis-estimating probabilities
\cite{rabin-law-small-numbers,tversky-law-small-numbers}.
Drawing on these results, 
a recent line of research has developed theoretical models of planning
by biased agents, seeking to bound the gap between the quality
of the plans produced by these biased agents and the quality of
optimal plans 
\cite{albers-time-inconsist,gravin-variable-present-bias,ko-ec14-full,kor-ec16,tang-time-inconsist}.

These analyses have generally considered a single bias at a time,
which serves as a way to decompose a complex pattern of behavior 
into a set of conceptually distinct parts.
But it is natural to ask what phenomena might emerge if we were to
build models of multiple biases acting at once.
Would they reinforce each other, or partially ``cancel each other out,''
or would it be situationally dependent?

In this paper we investigate the prospect of analyzing multiple
biases simultaneously, using a theoretical model as our underlying approach.
We focus on two well-studied behavioral biases that fit
naturally together: {\em present bias} --- the tendency to
value costs and benefits incurred in the present too highly relative to 
future costs and benefits 
\cite{akerlof-procrastination,pollak-time-inconsist,strotz-time-inconsist} --- 
and {\em sunk-cost bias} ---
the tendency to incorporate costs incurred in the past into one's
plans for the future, even when these past costs are no longer
relevant to optimal planning
\cite{thaler-positive-theory-choice,arkes1985psychology,thaler-mental-accounting}.
Sunk cost bias is a fundamental bias in planning, studied in various
  disciplines under different names. For example, it appears in the organizational behavior
  literature as
  ``Escalation of Commitment'' \cite{staw1976knee}, and it is known as the ``Concorde
  Fallacy'' \cite{dawkin1976selfish, weatherhead1979savannah} in behavioral
  ecology, named after the
famous supersonic airplane whose development was continued long after it was
clear that it had no economic justification.

\xhdr{Present Bias and Sunk-Cost Bias}
Present bias and sunk-cost bias on their own are qualitatively quite different,
though each operates on perceptions of costs and benefits over time, and
each is easily recognizable at an intuitive level.
A canonical example of present bias (familiar from many people's
experience) is the scenario in which an individual buys a membership to a gym,
but then actually goes to the gym very few times \cite{dellavigna2006paying}.
Viewed at the moment when the gym membership was purchased, the long-range
health benefits of regular exercise seemed to outweigh the cost in effort 
required to go to the gym regularly; but when the time comes to actually
go to the gym, the cost of the effort seems larger than it did previously, 
even relative to the other costs and benefits under consideration.
This leads to sub-optimal decision-making: either it would have been
preferable to buy the membership and then regularly go to the gym,
or to not buy the membership, but it can't be optimal to buy a membership
and then not use it.

A canonical example of sunk-cost bias (also familiar from everyday
experience) lies in the contrast between the following two scenarios
\cite{thaler-positive-theory-choice,thaler-mental-accounting}:
\begin{itemize}
\item[(i)] You have bought an expensive and
non-refundable ticket to a concert or sporting event that you
are very interested in attending, 
but on the day of the event, 
a major snowstorm makes travel dangerous.  Should you go to it anyway?
\item[(ii)] You were given a free ticket to a concert
or sporting event that you are very interested in attending, 
but on the day of the event, 
a major snowstorm makes travel dangerous.  Should you go to it anyway?
\end{itemize}
In examples of these and similar situations, many people view the
two scenarios differently --- they would risk the dangerous travel
conditions in scenario (i) so as not to ``throw away the cost of the ticket,''
while they'd conclude in scenario (ii) that it's not worth the risk
just to make it to the free event.
Yet if we think of the two scenarios strictly as an optimization of
costs and benefits, they are effectively equivalent: 
since the cost of the ticket is unrecoverable in scenario (i),
in both cases the question is whether the enjoyment of attending
the event (given that you are already in possession of the ticket)
outweighs the costs associated with traveling under risky conditions.
The fact that the two scenarios feel different at an intuitive level
suggests some of the deep
ways in which people take into account {\em sunk costs} --- costs incurred
in the past that can no longer be recovered --- and use these sunk costs
in their decision-making even when they are formally irrelevant to
the optimization aspects of the planning problem ahead.

\xhdr{Interactions of Present Bias and Sunk-Cost Bias}
Although present bias and sunk-cost bias involve different types of
reasoning, they both connect costs and benefits incurred at 
different stages of a planning problem to decisions about future behavior.
As such, one could ask about the behavior of an agent
in such a planning problem if they were experiencing both biases.
Do we learn something new by considering the two biases together?

We argue here that modeling the interaction of present bias and
sunk-cost bias in planning
leads to an interesting and natural set of phenomena that don't
arise when we model either of the two biases individually.
To get some intuition for what we learn by combining them, let's first
return to a synthesis of the two scenarios discussed above.
In particular, consider the reasoning (again familiar from
everyday life) of a person who decides they're
going to buy a gym membership so that when the time comes to go to the
gym, their desire not to waste the money spent on the membership will help
motivate them to go regularly.
Although the sentiment is expressed in a pithy format, it is 
intrinsically based on an interaction 
among multiple ingredients:
first, the person suffers from present bias, which will make it harder
to attend the gym when the time comes; 
second, they exhibit sunk cost bias so once they buy a gym membership
  they will be more inclined to visit the gym to avoid wasting the money they
  already spent;
and third, they are {\em sophisticated} in that they realize they
will experience these biases in the future, so they plan to use their sunk-cost bias associated with
prepaying the gym membership as a commitment
device to overcome their present bias when it arises.

\subsection{Planning with Multiple Biases: A Basic Model}

We now describe a simple theoretical model in which we can express
this type of planning, and discuss the basic framework for reasoning
within the model.
The model has the following components, building on a graph-theoretic
formalism from our prior work on present bias 
\cite{ko-ec14-full,kor-ec16}.
First, the planning problem is represented by a directed graph $G$
with non-negative costs on its edges.
An agent starts at a node $s$ in $G$ 
with the goal of reaching a node $t$ in $G$.
There is a reward $R$ at the node $t$.
The agent's payoff if it reaches $t$ is equal to the reward $R$ minus the 
sum of the costs on all the edges it traverses.
In case the agent traverses some of the edges but doesn't
	reach $t$ its negative payoff is simply the total cost of all edges traversed.
(The agent achieves a payoff of 0 if it never starts traversing
the graph.)

Figure \ref{fig:gym-example} shows a small instance of this type of
planning problem.  The optimal plan would be to traverse the
upper path through $v$, achieving a payoff of $R - 1 - 12 = 6$.
Note that if we set the reward $R$ to be 10 instead of 19, then
the optimal plan would be not to start, thus achieving a payoff of 0.

\begin{figure}[ht] \centering 
  \begin{tikzpicture}[->,shorten >=1pt,auto,node distance=2cm,
    thin, scale=.8] 
    \node (s) [circle, draw, minimum size=.8cm] at (-3,0) {$s$}; 
    \node (v) [circle, draw, minimum size=.8cm] at (0,1) {$v$}; 
    \node (w) [circle, draw, minimum size=.8cm] at (0,-1) {$w$}; 
    \node (t) [circle, draw, minimum size=.8cm] at (3,0) 	{$t$};
    \node (R) at (4.5,0) 	{$R = 19$};
    \path (s) edge node {$1$} (v)
    (v) edge node {$12$} (t) 
    (s) edge node {$4$} (w)
    (w) edge node {$10$} (t)
    ; 
  \end{tikzpicture}
 \caption{An instance of the planning problem.}
  \label{fig:gym-example}
\end{figure}
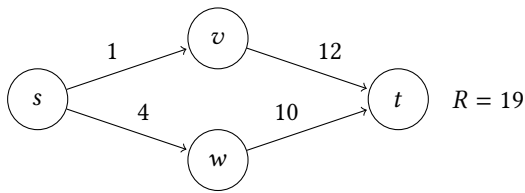

\xhdr{$\bm{(b,\lambda)}$-agents}
Now, let's consider how to model the behavior of biased agents on 
such a graph $G$.
Biased agents can deviate from optimal behavior in two ways: first
in how they misperceive the costs of paths in $G$ (and hence how
they misperceive payoffs), and second
in their potential misunderstanding of how they will behave in the future.
We describe these two components in turn, and then illustrate them by 
showing how biased agents behave in 
the example of Figure \ref{fig:gym-example}.

Our agents will exhibit both present bias and sunk-cost bias in general,
and we specify them using parameters $b \geq 1$ and $\lambda \geq 0$.
\begin{itemize}
\item The quantity $b$ is a present-bias parameter: when the agent is at a
node $u$ and considering the prospect of traversing a path $P$ beginning
with the edge $(u,v)$, 
it {\em perceives} the cost of the $u$-$v$ edge as being scaled up
by a multiplicative factor of $b$.
It adds this scaled-up cost to the actual costs of the remaining 
edges on $P$, resulting in a total {\em perceived cost} for $P$.
This reflects the overweighting of costs incurred in the present ---
in this case, the next edge to be traversed ---
that is associated with agents exhibiting this bias
\cite{laibson-quasi-hyperbolic}.\footnote{The model proposed by
\citet{laibson-quasi-hyperbolic} also include an exponential decay
on costs and rewards incurred in the future, where for a decay parameter
$\delta$, quantities experienced $\tau$ steps in the future are reduced by a 
factor of $\delta^\tau$.  In this paper we consider the case of 
$\delta = 1$, where there is no decay into the future, so as to focus
our attention on the present-bias parameter $b$.}
Thus, for example, an agent with present-bias parameter $b$ located at
$s$ in Figure \ref{fig:gym-example} would perceive the cost of the 
upper path as $b + 12$ and the lower path as $4b + 10$.
\item The quantity $\lambda$ is a sunk-cost parameter: if the agent
decides to {\em abandon} the traversal --- stopping at its current node $u$
and thus incurring no future cost or reward --- 
it exhibits a mental cost equal to $\lambda$ times the total cost it has incurred thus far.

This reflects the agent's aversion to giving up when it already
has incurred sunk cost in the traversal thus far;
incurring this as a final cost is motivated by constructions in the
literature on {\em mental accounting}
\cite{thaler-positive-theory-choice,thaler-mental-accounting}
and {\em realization utility}
\cite{barberis-realization-utility}.

\end{itemize}
We will refer to a biased agent with parameters $b$ and $\lambda$ as
a {\em $(b,\lambda)$-agent.}
Note that an agent who experiences neither present bias nor sunk-cost bias
has $b = 1$ and $\lambda = 0$, and hence is a $(1,0)$-agent.

\xhdr{Future Selves, Naivete, and Sophistication}
So far we have described the way a biased agent perceives costs;
now we need to describe the process by which it forms a plan on the
graph $G$.
Since our focus is on agents who may have different preferences
in the future than they do in the present,
we adopt a style of exposition used in behavioral economics and
consider agents who reason about what their ``future selves'' will do.
This style of description is useful in our gym membership scenarios, 
for example, where the person buying the gym membership would
like his or her ``future self'' to go to the gym regularly, but is worried
that this future self will not feel like going when the time comes
to actually do it.

This is also a useful formulation
for our graph-theoretic model, because biased agents may
differ in how they believe their future selves are going to behave.
Suppose a $(b,\lambda)$-agent is currently located at a node $u$,
and is considering whether to traverse an edge $(u,v)$.
It imagines that when it reaches the node $v$, it will hand off
control of future planning to its ``node-$v$ self.''
Now, how does the agent believe its node-$v$ self will reason about
the remainder of the planning problem?
An agent who is {\em naive} about its biases believes that its
node-$v$ self will plan optimally starting from node $v$,
whereas an agent who is {\em sophisticated} about its biases
believes that its node-$v$ self will continue to behave like a 
$(b,\lambda)$-agent.
Since in our model, an agent's parameters remain constant for the
duration of the planning problem, a sophisticated agent is correct
in its belief about its node-$v$ self, while a naive agent is incorrect
in its belief.
Importantly, both types of agents care about the costs incurred
by their future selves as well as their own costs; they just scale
up the cost of the immediate next edge by a factor of $b$ when they
determine the total cost of a path, reflecting the fact that they
value costs to themselves a factor of $b$ higher than they value costs
to their future selves.

There is extensive empirical evidence that people can behave more like 
naive agents or more like sophisticated agents in different scenarios ---
sometimes we make a plan believing that we'll be fully motivated
to follow through on it when the time comes, and sometimes we 
factor into our planning the belief that we might not be inclined
to take the necessary step in the future
\cite{frederick-time-inconsist-surv}.

There are thus multiple types of agents, and as we will see next,
they exhibit a range of intuitively natural behaviors that reflect
how their biases --- and their awareness of these biases --- interact.
In keeping with the fact that there are two kinds of biases under consideration,
we will refer to the two types of agents discussed above as
{\em doubly naive} and {\em doubly sophisticated}, indicating that
such agents are either naive about both biases or sophisticated about
both biases.  Later we will consider the natural question
of agents who are naive about one bias and sophisticated
about the other.

\subsection{Two Examples}

With the core definitions established, it is very useful to consider
the behavior of these agents in some basic examples, for two reasons.
First, given the subtle distinctions among different agent types, it
is useful to see these distinctions through simple illustrations;
and second, these examples help establish that the behaviors
we are modeling are all intuitively quite natural.

\xhdr{Health Club Memberships}
We'll start with the behavior of these agents on the instance in 
Figure \ref{fig:gym-example}.
To begin with, we note that the planning problem described by
the graph in Figure \ref{fig:gym-example} has a direct interpretation
in terms of decisions about gym membership, as follows.
\begin{quote}
A local health club offers a range of classes, and you're interested
in taking its yoga class.
The effort required to take the yoga class is 10, and the long-term
reward from having taken it is 19.
To take the yoga class you need to get a membership at the health club,
and there are two options for memberships.  
With a {\em basic membership}, you pay 1 up front, and pay 2 for each class 
at the time you attend it.
With a {\em deluxe membership}, you pay 4 up front, but then all classes
are free.
You know that you only want to take the yoga class, not any
of the other classes (since none of the
other classes at the health club appeal to you).
What should you do?
\end{quote}
It is easy to check that the graph in Figure \ref{fig:gym-example} encodes
this story, with node $v$ corresponding to the state in which you've
purchased a basic membership (but haven't yet taken the class),
node $w$ corresponding to the state in which you've
purchased a deluxe membership (but haven't yet taken the class),
and node $t$ corresponding to the state in which you've completed
the yoga class (and so can now achieve the reward of 19).
An optimal agent would buy the basic membership (using the path through $v$),
since there's no reason to pay 4 for a deluxe membership when
just the yoga class can be taken for a cost of $1 + 12 = 13$.

Now, what would doubly naive or a doubly sophisticated agent
do in this situation?
For concreteness, let's use $b = 2$ and $\lambda = 1/2$ as the parameters
for our example.  In particular, this means that both types of
biased agents --- doubly naive and doubly sophisticated ---
will perceive the cost on the first edge out of $s$ as being
multiplied by a factor of $2$; they differ in how they reason about
the remainder of the planning problem.
As part of this reasoning, it is important to distinguish between 
an agent's {\em perceived payoff} at a given point in the traversal, and the
actual payoff it incurs, which is determined entirely by the true
costs and rewards on the graph, rather by the agent's biases.
\begin{itemize}
\item A doubly naive agent would perceive the path through $v$ as costing
$2 \cdot 1 + 12 = 14$ and the path through $w$ 
as costing $2 \cdot 4 + 10 = 18$.
It also believes that it will behave optimally starting from 
whichever node it visits next.
It thus traverses the edge from $s$ to $v$.
Once it is at $v$, however, it now evaluates the cost of the $v$-$t$ edge
as $2 \cdot 12 = 24$, which means that paying this cost to get the
reward of 19 leads to a perceived payoff of -5.
On the other hand, abandoning the path incurs a sunk cost penalty of
$1 \cdot \lambda = 1/2$; since this perceived payoff of $-1/2$ is preferable to
the perceived payoff of $-5$ from continuing, 
the agent abandons the path at $v$.
In summary: the doubly naive agent buys the basic membership, but when
the time comes to take the yoga class, it lets the membership go to waste.
\item A doubly sophisticated agent first reasons about how it
would expect to behave starting from node $v$ and from node $w$.
From node $v$, with a sunk cost of $1$, it would behave the way the
doubly naive agent actually behaved when it reached $v$ --- comparing
a perceived payoff of $-5$ from continuing 
with a perceived payoff of $-1/2$ from abandoning ---
and so it would abandon the path if it were at node $v$.
From node $w$, with a sunk cost of $4$, it would get a perceived payoff of 
$19 - 2 \cdot 10 = -1$ from continuing to $t$, and a perceived payoff of 
$- 4 \lambda = -2$ (from the sunk-cost penalty of $2$) if it were
to abandon the path at $w$.  Thus, from $w$ it would continue to $t$.
Finally, back at $s$, the agent reasons that the path through $w$ to $t$
has perceived cost $2 \cdot 4 + 10 = 18$ 
and hence a perceived payoff of $19 - 18 = 1$
to its present self, since it knows it will continue from $w$.
Thus it chooses to go to $w$.
The informal summary is that the doubly sophisticated agent buys
the deluxe membership, since it knows the fear of wasting the price of the
deluxe membership (as manifested through its sunk-cost bias)
will motivate it to take the yoga class, leading to a positive payoff.
\end{itemize}
The upshot is that the optimal agent, the doubly naive agent, and the
doubly sophisticated agent all pursue different plans: the optimal
agent makes effective use of the basic membership; the doubly naive
agent foolishly buys the basic membership and then doesn't actually take
the yoga class; and the doubly sophisticated agent buys the
deluxe membership as a commitment device to follow through on the yoga class.

It is also instructive to compare these outcomes to the plans pursued
by naive and sophisticated present-biased agents --- that is,
$(b,0)$-agents who experience only present bias
without sunk-cost bias.
A naive present-biased agent will follow the same plan as the doubly naive
agent above.
But a sophisticated present-biased agent will follow a plan distinct 
from all the ones we've seen so far: it will correctly recognize that it 
wouldn't continue from either node $v$ or node $w$, and consequently
it wouldn't start out from $s$.
In other words, a sophisticated present-biased agent wouldn't buy either
type of membership in the health club, because it realizes that it won't
take the yoga class when the time comes.

\xhdr{Completing Assignments in a Class}
We briefly consider a second example where the contrasts between the agents
turn out differently --- a version of an example 
from \cite{ko-ec14-full} involving assigned work in a class, adapted
to the types of agents we are considering here.

Suppose you're taking a 4-week class, and you must complete three 
short projects by the end of the class.  In each week you can choose 
to do 0, 1, or 2 of the projects; doing 0 projects in a given week
costs 0, doing 1 project costs 4, and doing 2 projects costs 10.
If you complete all three projects by the end of the 4 weeks, then you
pass the class, which comes with a reward of $R = 17.5$.
The graph $G$ associated with this planning problem is shown in
Figure \ref{fig:deadline-example}: the node $v_{ij}$ corresponds
to the state in which you're $i$ weeks into the class and you've
completed $j$ projects so far.

\begin{figure}[h]
\begin{center}
\includegraphics[width=3.50in]{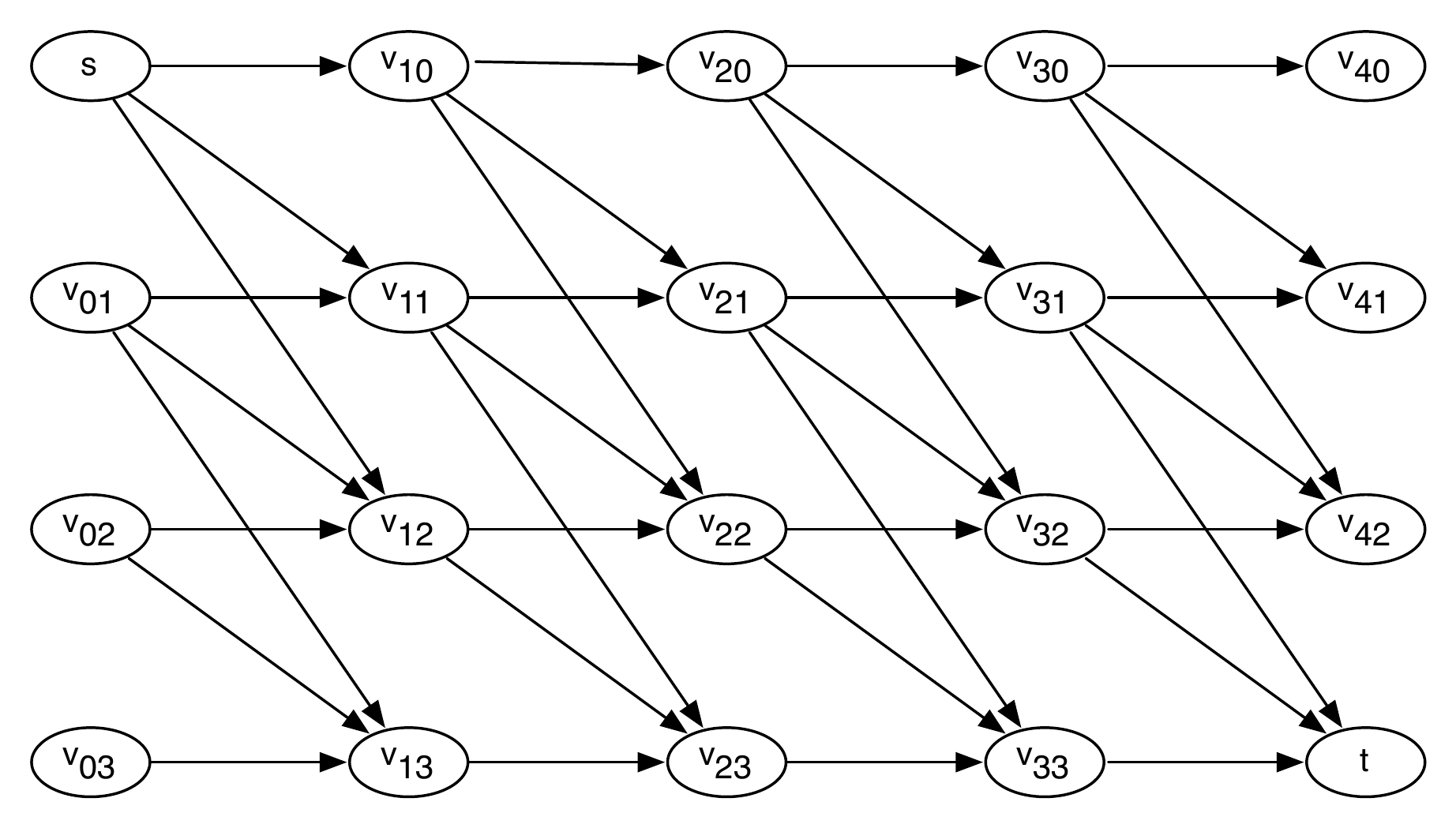}
\end{center}
\caption{A biased agent must choose a path from $s$ to $t$.}
\label{fig:deadline-example}
\end{figure}

An optimal agent would choose to do one project in each of 3 separate weeks,
for any 3 out of the 4 weeks,
incurring a total cost of 12 and hence a payoff of $17.5 - 12 = 5.5.$
Let's consider the behavior of biased agents with $b = 2$ and $\lambda = 3/4$;
we only sketch the reasoning for this example.
\begin{itemize}
\item A doubly naive agent will do one project in week 2, planning to do
one more project in each of weeks 3 and 4.  In week 3, it chooses to
defer both projects to week 4 (since $2 \cdot 4 + 4 > 10$).  In week 4,
it would incur a perceived payoff (due to sunk cost) of $-4 \lambda = -3$ from 
abandoning, and a perceived payoff 
of $17.5 - 2 \cdot 10 = -2.5$ from continuing,
so it will do both projects and finish the class.
\item A doubly sophisticated agent correctly anticipates that it will do
two projects in week 4, for a cost of 10.  Thus, when the time comes to do
the first project, its perceived cost will be $2 \cdot 4 + 10 = 18 > 17.5$,
and since it has no sunk cost at this point, it will choose to abandon
the path.  Given this, a doubly sophisticated will choose not to start out
from $s$, thus deciding not to take the class.
(Essentially, the doubly sophisticated agent says, correctly,
``I know that once I put some work into the class, I'm going to end up
pushing the rest of the work to the very end and overdo things in week 4.'')
\item A naive agent with only present bias --- i.e. a naive $(2,0)$-agent ---
will, like the doubly naive agent, get to week 4 needing to do two projects.
At that point, since it has no sunk cost bias, abandoning the path 
has a payoff of 0, 
while continuing has a perceived payoff of $17.5 - 2 \cdot 10 = -2.5$.
Thus it will abandon the path (dropping the class) in week 4.
\item A sophisticated agent with only present bias --- 
i.e. a sophisticated $(2,0)$-agent --- actually behaves optimally.
It correctly anticipates that if it reaches week 4 with two projects
left to do, it will abandon the path (since it has no sunk-cost bias),
and so it does one project in each of weeks 2, 3, and 4.
\end{itemize}

It's interesting that in contrast to the case of health club memberships,
where the doubly sophisticated agent reached the goal and the 
sophisticated present-biased agent didn't, here the roles are reversed;
the contrast is that this second example is one in which 
(i) the doubly sophisticated agent realizes that its present bias combined
with its sunk-cost bias will lead it down a path 
where it pays too high a price; and
(ii) there's an alternate path that the sophisticated present-biased
agent can take.

One conclusion from all these examples is that relatively small graphs
can encode scenarios that would otherwise be quite complicated to
reason about; and the interplay between present bias and sunk-cost bias
in these examples is producing intuitively natural
behaviors that inherently require both biases.

\subsection{Mixed Forms of Sophistication}

Since we are considering two biases at once, we should also consider the
possibility that an agent might be naive about one of its two biases
and sophisticated about the other.
Thus, a $(b,\lambda)$-agent at a node $u$, considering the traversal
of edge $(u,v)$, would be naive about its sunk-cost bias but sophisticated
about its present bias if it believed that its node-$v$ self will behave
like a $(b,0)$-agent who is sophisticated about its present bias.
Alternately, it would be naive about its present bias but sophisticated
about its sunk-cost bias if it believed that its node-$v$ self will behave
like a $(1,\lambda)$-agent who is sophisticated about its sunk-cost bias.

It is not hard to show that
in our model this latter type of agent, who is naive 
about present bias and sophisticated about sunk-cost bias, is indistinguishable
in its behavior from an agent who is naive about both
biases.\footnote{To see why this is the case, first observe that
an agent who has no present bias (i.e. $b = 1$) will behave the same
in our graph traversal problem regardless of whether or not
it has sunk-cost bias, and whether or not it is aware of it.
Since an agent who is naive about present bias plans paths on
the assumption that it will have $b = 1$ in the future, the plan
it makes from any node is indistinguishable from the plan of an agent
who is naive about both biases.}
Thus, we will focus on agents that are sophisticated about their present
bias and naive about their sunk-cost bias. We refer to such agents as {\em
singly sophisticated} agents.

It is interesting to consider how a singly sophisticated agent behaves
in the examples from Figure \ref{fig:gym-example} and
\ref{fig:deadline-example}.
In the health club example of Figure \ref{fig:gym-example}, the singly
sophisticated agent doesn't appreciate that sunk-cost bias will play
a role in its reasoning at nodes $v$ and $w$, and so at node $s$ it reasons
like a sophisticated $(b,0)$-agent and decides not to start out from $s$.
In the class-projects example of Figure \ref{fig:deadline-example},
the singly sophisticated agent again starts out reasoning like a
sophisticated $(b,0)$-agent and plans to do one project in 
each of weeks 2, 3, and 4.
Once it does the first project in week 2, however, it now has acquired
some sunk cost, and so it changes its plan to do both of the remaining
projects in week 4.
Thanks to the sunk-cost bias it goes ahead and does this,
since the perceived payoff of
$17.5 - 2 \cdot 10$ from finishing in week 4 is preferable to the 
perceived payoff of $-4 \lambda = -3$ from abandoning in week 4.
This second example shows that despite the agent's sophistication about
its present bias, its naivete about its sunk-cost bias means that
it can still sometimes be time-inconsistent in its behavior, changing
plans in the middle of its traversal of the graph $G$.

\xhdr{Perceived Rewards} 
We now describe an equivalent way of representing the payoff to 
an agent with sunk-cost bias.
Thus far, if a $(b,\lambda)$-agent stops without reaching $t$,
it incurs a negative payoff equal to $\lambda$ multiplied by 
the total cost of edges it has traversed.
In deciding whether to continue, it compares this negative payoff
from stopping to the perceived payoff from continuing to $t$
(equal to the reward $R$ minus the perceived cost of upcoming edges).
An equivalent way to express this comparison is to add $\lambda$ times
the cost incurred so far to the reward, creating a new (larger)
{\em perceived reward}.  The agent continues if and only if this
perceived reward is at least as large as the perceived cost of the
upcoming edges it plans to traverse.  In this way, there is no 
explicit sunk-cost penalty from stopping; rather, the sunk-cost
bias is reflected in the growing reward that the agent perceives,
incorporating $\lambda$ times the cost experienced so far.
We will use this equivalent formulation in the remainder of the paper.

\subsection{Overview of Results}

In the remainder of the paper, we provide a set of performance guarantees
and algorithmic results for the types of biased agents defined 
in this section.  
We give a brief summary of some of the main results here.

We first consider doubly sophisticated agents, and in particular the planning
problem for such agents.  Algorithmically, such agents face a non-trivial
planning task, since in choosing a next step, they
must consider what their future selves will do not just from every
node, but for every possible value of the sunk cost they might
experience from that node.
We give an algorithm for solving the planning problem for
doubly sophisticated agents that runs in time polynomial in the number
of nodes $n$ and the total sum $C$ of edge costs in the graph.
This is a {\em pseudo-polynomial} algorithm in that its
running time depends on the actual magnitudes of the costs in the instance,
and it is natural to ask whether there might be some 
better algorithm that avoids this form of dependence on the costs.
We show, however, that this dependence is necessary (assuming $P \neq NP$),
by proving that the planning for a doubly sophisticated agent is
NP-hard when the edge costs are presented in binary notation
(and hence the input has size polynomial in $n$ and $\log C$).

In a positive direction, 
we are able to show that doubly sophisticated
agents always achieve reasonably good payoffs.
In particular, we find that if $C_o$ denotes the cost 
incurred by an optimal agent that reaches $t$ in a given instance,
then the payoff of a doubly sophisticated agent is 
smaller than the payoff of an optimal agent by
an additive amount of at most $(b-1) C_o$.
As one direct consequence of this fact, a doubly sophisticated
agent will reach the target node $t$ in any instance for which
the reward $R$ is at least $b C_o$.
We show similar additive gaps between the payoff of a doubly
sophisticated agent and the payoff of a sophisticated present-biased agent
(with no sunk-cost bias, and hence parameters $(b,0)$); there are
instances in which either can achieve a better payoff than the other,
but the gap between them always remains bounded by $(b-1) C_o$.

For doubly naive agents, we show that their sunk-cost bias
can push them to incur costs that are much higher than the available
reward $R$.  In particular, they can incur a cost that
is exponential in the size of the graph.
We find upper and lower bounds on the worst-case cost, with
exponential bases that are close to one another between the two bounds.

Using a more complex construction, we can show that exponentially
bad bounds apply to the singly sophisticated agent as well.
Despite its sophistication about its present bias, it is possible
for a singly sophisticated agent to incur exponentially large cost
before abandoning the traversal without reaching $t$.
We complement this with a nearly matching upper bound, which also
shows that the cost incurred by a singly sophisticated agent is only
exponential in the number of ``switches'' --- 
nodes at which the agent changes its plan.

\section{Doubly Sophisticated Agents}  \label{SEC:DOUBLY-SOPH}
In this section we consider doubly
sophisticated agents. Recall that these are agents that are sophisticated about both their
present bias and sunk cost bias. A doubly sophisticated
agent accurately predicts the decisions that its future selves will make, meaning that the agent will follow the path it plans to take. In particular, this means that the agent won't begin traversing the graph unless it is sure that it will reach the target. 

Path-planning for an agent that is sophisticated but has no sunk cost bias
is
straightforward -- at a node $v$, the agent's action is purely a function of its
decisions at later nodes in the topological ordering, so its decisions can be
recursively computed. With sunk cost bias, however, this is no longer the case.
An agent's decision depends not only on its future decisions but also on its past decisions
and particularly on the cost it has incurred reaching $v$. Thus, to plan its path it needs 
to know its future behavior for all possible values of the cost incurred.

In the next section we will see that when the number of possible values of the cost incurred at every node is small the agent can efficiently recursively compute the path it will take and discuss the special case in which the cost on the edges have integer values.

\subsection{Integer Doubly-Sophisticated Path Computation}

    As we will later see the general path computation problem for a doubly sophisticated agent
	is NP-hard. Let $k$ be an upper bound on the number of possible different values of the costs for reaching a node.  \full{In Appendix \ref{app:rec-alg}} \camera{In the full version} we present a recursive algorithm for path computation that runs in time polynomial in $k$ and $n$. Here we present an iterative dynamic program algorithm for the case that the edges have integer costs. For this case we take $k$ to be the sum of all edge costs and exhibit a pseudo polynomial algorithm. Such an algorithm is both easier to follow and illustrates well the way that a doubly sophisticated agent reasons about the behavior of its future selves to plan its path. 

\begin{proposition}
	The integer doubly sophisticated path computation problem can be solved in
	time polynomial in $n$ and $C$, where $n$ is the number of vertices in $G$ and
	$C$ is the sum of the costs of the edges.
\end{proposition}
\begin{proof}
	
	\aref{intds} solves the integer doubly-sophisticated path computation problem
  in time polynomial in $n$ and $C$. The algorithm relies on the observation
  that path that the agent will choose from some node $u$ only depends on the
  total cost of the path it took to get to $u$ and not on the explicit path.
  Since the costs on the edges are integers, we can compute for each node $u$
  and for each possible cost of the path reaching $u$ what the cost of going
  from $u$ to $t$ is. This can be done in reverse topological order -- at node
  $u$ with some sunk cost $i$, if we know what the agent will do at every
  subsequent node $v$ with sunk cost $i + c(u, v)$, we can determine the
  agent's behavior at $u$.
	
  In the algorithm we define two arrays -- $choices$ and $costs$ -- that hold the
  choice and cost of the path that the
	agent would take if it reached a vertex $u$ with a sunk cost of $i$. We begin
  filling in these arrays in reverse topological order, since at $t$, there is
  no decision to be made, and the cost of the remaining path is 0.

  If the agent reaches some vertex $u$ having incurred
	a cost of $i$ along the way, then its choice of where to go next is uniquely
	determined by the successor vertices it can go to. Since we do the computation
	in reverse topological order, we know that for any successor $v'$, the
	$choices$ and $costs$ values have already been computed. Therefore, the agent
	can simulate what would happen along each potential path simply by looking up
	the $costs$ value of reaching $v'$ with an incurred cost of $i + c(u,v')$.
  Out of all potential successors $v'$, the agent chooses the $v$ that minimizes
  its perceived cost, where perceived cost is given by
	$b \cdot c(u,v') + costs[v][i+c(u,v)] $.

  If this
	perceived cost is larger than the perceived reward, which is given by
  $R$ plus $\lambda$ times the cost incurred so far, then the agent would abandon upon
	reaching $u$ after an incurred cost of $i$. Otherwise, it would proceed to
	$v$, and the total cost of the path it would take from $u$ to $v$ would be
  $c(u,v) + costs[v]$. In either case, the algorithm correctly computes the
  action of the agent.
	
	At the start of the traversal, the agent is at $s$ with an incoming cost of 0.
  Therefore, we can look up $choices[s][0]$ to see where the agent would
	go next, and so on until we find the path that the agent would take to $t$,
	updating the incurred cost so far as we go.
\end{proof}

\alg{\textsc{IntegerDoublySophisticated}$(G, R, b, \lambda)$}{intds}
\State $n \gets |V|$
\State $C \gets$ sum of the edge costs
\State $choices \gets array[n][C]$ initialized to $null$
\State $costs \gets array[n][C]$ initialized to $0$
\For{$u \in V \backslash \{t\}$ in reverse topological order}
\For{$i \gets 0 \dots C$}
\State $v \gets \argmin_{v' \in N(u)} b \cdot c(u,v') + costs[v'][i+c(u,v')]$
\State $perceived \gets b \cdot c(u,v) + costs[v][i+c(u,v)]$
\If {$perceived > R + \lambda \cdot i}$
\State $choices[u][i] \gets null$
\State $costs[u][i] \gets \infty$
\Else
\State $choices[u][i] \gets v$
\State $costs[u][i] \gets c(u,v) + costs[v][i+c(u,v)]$
\EndIf
\EndFor
\EndFor
\If {$choices[s][0] == null$}
\State \Return no path
\EndIf
\State $path \gets []$
\State $cost \gets 0$
\State $u \gets s$
\While {$u \ne t$}
\State Append $u$ to $path$
\State $v \gets choices[u][cost]$
\State $cost \gets cost + c(u, v)$
\State $u \gets v$
\EndWhile
\State Append $t$ to $path$
\ealg

\subsection{The Gap Between a Doubly Sophisticated Agent and an Optimal Agent}
As the payoff of a doubly sophisticated agent is always non-negative, the only
instances that can admit a positive gap are ones in which the optimal agent
reaches the target. Let $C_o(u)$ be the cost of the optimal agent for reaching
the target from $u$, which means that $C_o(s) = C_o $. We show that there can be
an additive gap of at most $(b-1)C_o$ between the payoffs of an optimal agent and a doubly sophisticated agent. We note that the source of the gap could be either because the doubly sophisticated agent did not traverse the graph or because both agents traversed the graph but the cost of the doubly sophisticated agent was higher. 

Instead of proving the gap directly we show that a similar claim holds in a more general setting. An agent currently at $v$ that exhibits sunk cost bias perceives a different reward
based on the path it took to $v$. In particular if the agent took a path $P$ to get to $v$ then its perceived reward at $v$ is $R + \lambda \cdot c(P)$, where $c(P)$ is the
total cost of the path $P$. To generalize this, we define a reward schedule $H$ as a mapping from paths beginning at $s$ to
rewards. When computing a path for a graph $G$ with reward schedule $H$, an
agent makes its calculations as if after following a path, the reward it
will get when it reaches $t$ is given by the reward schedule.

\begin{proposition} \label{prop:general_cost_ratio}
		Given a graph $G$ and a path-dependent reward schedule $H$, if the perceived reward according to $H$ at each vertex $v$ on the optimal path from $s$
		to $t$ is at least $b \cdot C_o(v)$, then a present-bias sophisticated agent will traverse the graph and incur a cost of at most $b \cdot C_o(s)$ .
\end{proposition}
\begin{proof}
	Let $P$ be the optimal path from $s$ to $t$. We will prove by induction that from each vertex $v$ along $P$, there exists a path of cost at most $b \cdot C_o(v)$ from $v$ to $t$ that the agent
	would be willing to take for the reward schedule $H$, regardless of the path
	taken to get to $v$. Let $C_H(v)$ denote the cost for a sophisticated
  present-biased agent to reach $t$ from $v$ given the schedule $H$.
	
	\textbf{Base case:} At $t$, the claim is trivially true.
	
	\textbf{Inductive hypothesis:} $C_H(v) \le b \cdot C_o(v)$, and $v$ is never
	abandoned under $H$.
	
	\textbf{Inductive step:} Consider some vertex $u$ on $P$ and assume that the inductive hypothesis holds for all the vertices after $u$ on $P$. Let $v$ be the next
	vertex after $u$ on $P$. By assumption, the perceived reward at $u$ is some $R_u \ge b \cdot
	C_o(u)$. We know by induction that if the agent reaches $v$, the rest of the
	path will cost $C_H(v) \le b \cdot C_o(v)$. Therefore, the perceived cost of
	going from $u$ to $v$ and then from $v$ to $t$ is
	\begin{align*}
	b \cdot c(u,v) + C_H(v) \le b \cdot c(u,v) + b \cdot C_o(v) = b \cdot C_o(u) \le R_u
	\end{align*}
	Thus, the agent would be willing to take this path, so $u$ could never be
	abandoned. Furthermore, this implies that the agent will take some path 
	(either the one discussed above or a different one)
	that its perceived cost is at most $b \cdot C_o(u)$. As the total cost is 
	always smaller than the perceived cost this implies that the total cost 
	of the path that the agent will take is at most $b \cdot C_o(u)$ as required. 
\end{proof}

We can now use Proposition \ref{prop:general_cost_ratio} to bound the gap between an optimal agent and a doubly sophisticated agent.
\begin{proposition} \label{prop:double_cost_ratio}
	Consider a task graph $G$ with a reward $R$ on the target. 
	The payoff of an optimal agent can be higher than the payoff
	of a $(b,\lambda)$-doubly sophisticated agent by an additive amount of at most $(b-1)C_o$.
\end{proposition}

\begin{proof}
	First, observe that if $R \leq b \cdot C_o$ then the payoff of the optimal agent is at most $(b-1) \cdot C_o$ and the proposition holds. Next, recall that the reward schedule that is used 
	to describe the behavior of a doubly sophisticated agent is monotonically increasing: the perceived reward at any
	vertex $v$ along a path $P$ is $R + \lambda \cdot C_P(v)$, where $C_P(v)$ is the
	cost for reaching $v$ on the path $P$. This means that if $R\geq b \cdot C_o$ we can
  apply Proposition \ref{prop:general_cost_ratio} and get that the doubly
  sophisticated agent will reach the target and pay a cost of at most $b \cdot
  C_o$. Hence, in this case as well, the difference in the payoffs of the optimal agent and doubly sophisticated agents is at most $ b \cdot C_o -  C_o = (b-1) C_o$.
\end{proof}
	
Lastly the fact that the agent is always willing to traverse the graph for a reward of $R=b \cdot C_o$ leads us to the following corollary: 

\begin{corollary} \label{cor:dsrewardratio}
	The minimum reward $R$ for which a $(b,\lambda)-$doubly sophisticated agent would be willing to traverse a graph $G$ is at most $b \cdot C_o$.
\end{corollary}

In Section \ref{section:ds-spb-comp} we present results for similar comparisons
between doubly sophisticated and sophisticated present-biased agents, finally in Section \ref{sec-doubly-np} we show that computing the path that a doubly sophisticated agent takes is NP-hard.

\subsection{Doubly Sophisticated Agents Versus Sophisticated Present-Biased
	Agents} \label{section:ds-spb-comp}
To better understand the interplay between present bias and sunk cost
bias, in this section we contrast between doubly sophisticated agent and
sophisticated present-biased agents. By Proposition
\ref{prop:double_cost_ratio} and \ref{prop:general_cost_ratio} we have that
for each of the agents the additive gap between its payoff and the payoff of
an optimal agent is at most $(b-1)C_o$ \footnote{For a present-bias sophisticated agent this gap was
	first proven in \cite{kor-ec16}.}. As the payoff of each of the agents is
at most the payoff of an optimal agent, we have that the gap between the
payoffs of a doubly sophisticated agent and sophisticated present-biased
agents is $(b-1)$. This proves the following claim:
\begin{claim} 
	The additive gap between the payoffs of a doubly sophisticated agent and a
	sophisticated present-biased agent is at most $(b-1)C_o$. 
\end{claim}

In the next two claims we will see that this gap can go either way and it is tight in both directions. In other words, each of the type
of agents can do better than the other by this additive factor of
$(b-1)C_o$. We conclude that the way the two biases interact with one
another in agents that are sophisticated about them depends on the situation.

\begin{claim} 
	The payoff of a doubly sophisticated agent can be smaller than the payoff of a
	sophisticated present-biased agent by an additive amount arbitrarily close to $(b-1)C_o$. 
\end{claim}
\begin{proof}
	Consider the example in Figure \ref{fig:sing-better} with $R=b^2-\lambda \eps$. A doubly sophisticated
	agent knows that after traversing the edge $(s,v_1)$ its sunk cost would
	increase the perceived reward to $b^2$ and it will be able to traverse the
	edge $(v_1,t)$. Since when standing at $s$ the upper path has a lower
	perceived cost, the doubly sophisticated agent will choose it for a total cost
	of $b+\eps$. A
	sophisticated present-biased agent, on the other hand, knows that it won't be
	able to traverse the edge $(v_1,t)$ and thus chooses the lower path of total
	cost $1+(b+1)\eps$ instead. Intuitively, the sunk cost bias is allowing the
	doubly sophisticated agent to ``procrastinate'' more because it knows that if
	it puts in a small amount of work now, its future self won't abandon because
	the perceived reward will be higher.
\end{proof}
\begin{figure}[ht] \centering 
	\begin{tikzpicture}[->,shorten >=1pt,auto,node distance=2cm,
	thin] 
	\node (s) [circle, draw, minimum size=1cm] at (-3,0) {$s$}; 
	\node (v1) [circle, draw, minimum size=1cm] at (0,1) {$v_1$}; 
	\node (v2) [circle, draw, minimum size=1cm] at (0,-1) {$v_2$}; 
	\node (t) [circle, draw, minimum size=1cm] at (3,0) 	{$t$};
	
	\path (s) edge node {$\eps$} (v1) 
	(s) edge node {$1$} (v2)  
	(v1) edge node {$b$} (t) 
	(v2) edge node {$(b+1)\eps$} (t) ; 
	\end{tikzpicture}
	\caption{For $R=b^2-\lambda \eps$ a doubly sophisticated agent will take the upper path and a singly sophisticated agent will take the lower path } \label{fig:sing-better}
\end{figure}
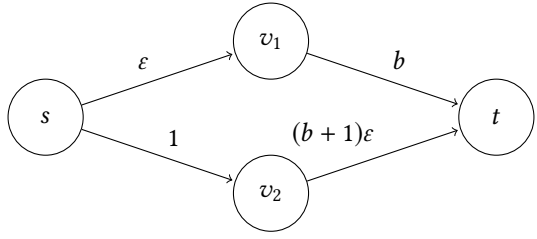

\begin{claim} 
	The payoff of a sophisticated present-biased agent can be smaller than the payoff of a
	doubly sophisticated agent by an additive amount arbitrarily close to $(b-1)C_o$. 
\end{claim}

\begin{proof} 
	Consider the example in Figure \ref{fig:doubly-vs-soph} with
	$R=b^2-\lambda \cdot \eps$. A sophisticated present-biased agent will not start
	traversing the graph as the perceived cost when standing at $v_1$ is greater
	than the reward. Hence, a sophisticated agent will incur a payoff of $0$. A
	doubly sophisticated agent, knows that because of sunk cost, the perceived
	reward when standing at $v_1$ will be sufficiently large to continue to the
	target. Thus, it will traverse the graph for a total payoff of $b^2-\lambda \cdot
	\eps - b -\eps = (b-1) \cdot (b+\eps) + (b-\lambda) \eps$.  
\end{proof}
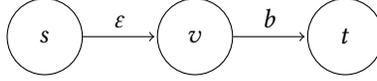
\begin{figure}[ht] 
	\centering 
	\begin{tikzpicture}[->,shorten >=1pt,auto,node distance=2cm, thin] 
	\node (s) [circle, draw, minimum size=1cm] at (0,0) {$s$}; 
	\node (v) [circle, draw, minimum size=1cm] at (2,0) {$v$}; 
	\node (t) [circle, draw, minimum size=1cm] at (4,0) {$t$};
	
	\path (s) edge node {$\eps$} (v) 
	(v) edge node {$b$} (t) ; 
	\end{tikzpicture} 
	\caption{For $R=b^2-\lambda \cdot \eps$, a
		sophisticated agent wouldn't traverse the graph but a doubly sophisticated agent would.}
	\label{fig:doubly-vs-soph}
\end{figure}

\subsection{Doubly-Sophisticated Path Computation is NP-Hard} \label{sec-doubly-np}
We will show that the problem of determining whether there exists a path that a doubly sophisticated agent will traverse is NP-hard. Formally, we show
\begin{theorem} \label{thm:doubly-hard}
This problem is NP-hard: Given a graph $G$, a reward $R$, present bias
$b$, and sunk cost bias $\lambda$, determine whether there exists a
path that a doubly sophisticated agent can take from $s$ to $t$.
\end{theorem} 

We show that for any parameter $1/2 \leq \lambda<1$ there exists $b>0$ such that
 this problem is NP-hard by using a reduction from the Subset
Sum problem. Recall that in the Subset
Sum problem we ask, given a set $S$ of integers $x_1, \dots, x_n$ and a target $T$, is
there some subset of $S$ that adds up to $T$. Given an instance of the subset sum problem, we construct a graph $G$ as follows:
The vertices are $s, v_1, \dots, v_{n+1}, w_1, \dots, w_n, t$. There is an edge
of cost $0$ from $s$ to $v_1$ and an edge of cost $T$ from $v_{n+1}$ to $t$. For
each $i$ between 1 and $n$, there is an edge of cost 0 from $v_i$ to $w_i$ and
an edge of cost 0 from $w_i$ to $v_{i+1}$. Finally from each $v_i$ to $v_{i+1}$,
there is a sequence of vertices and edges such that the path from $v_i$ to
$v_{i+1}$ has total cost $x_i$. The first two edges on each sequence have cost
$\frac{1}{2b}$, and each subsequent edge has twice the cost of the previous
edge. This sequence ends when the total cost of the sequence is exactly $x_i$
(meaning that the last edge has at most twice the cost of the second-to-last
edge, and the sum of the costs of the edges is $x_i$). For example, if $x_i=4$ and $b=2.5$,
then the sequence of edges will be $\frac{1}{5}, \frac{1}{5},
\frac{2}{5}, \frac{4}{5}, \frac{8}{5}, \frac{4}{5}$.  The agent's present bias for $\lambda>0$ is $b=2+\lambda$ and the reward is $R = (b-\lambda) T + \lambda -\eps = 2T + \lambda -\eps$. 
A sketch of the reduction can
be found in Figure \ref{fig:reduction}

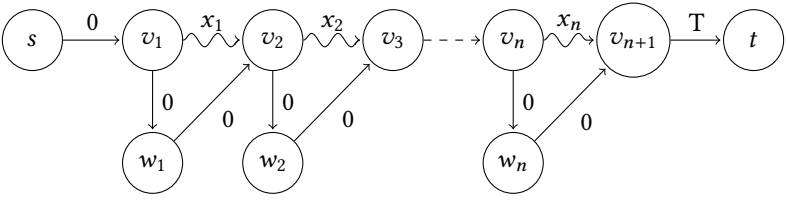
\begin{figure}[ht]
	\centering
	\begin{tikzpicture}[->,shorten >=1pt,auto,node distance=2cm, thin, scale=.8]
	\node (s) [circle, draw, minimum size=.8cm] at (-6,0) {$s$};
	\node (v1) [circle, draw, minimum size=.8cm] at (-4,0) {$v_1$};
	\node (v2) [circle, draw, minimum size=.8cm] at (-2,0) {$v_2$};
	\node (v3) [circle, draw, minimum size=.8cm] at (0,0) {$v_3$};
	\node (vn) [circle, draw, minimum size=.8cm] at (2,0) {$v_n$};
	\node (vn1) [circle, draw, minimum size=.8cm] at (4,0) {$v_{n+1}$};
	\node [below = .8cm of v1] (w1) [circle, draw, minimum size=.8cm] {$w_1$};
	\node [below = .8cm of v2] (w2) [circle, draw, minimum size=.8cm] {$w_2$};
	\node [below = .8cm of vn] (wn) [circle, draw, minimum size=.8cm] {$w_n$};
	\node (t) [circle, draw, minimum size=.8cm] at (6,0) {$t$};
	
	\path [->,decoration=snake]
	(v1) edge [decorate] node [above] {$x_1$} (v2)
	(v2) edge [decorate] node [above] {$x_2$} (v3)
	(vn) edge [decorate] node [above] {$x_n$} (vn1)
	;
	
	\path
	(v3) edge[dashed] node {} (vn)
	;
	\path
	(s) edge node {0} (v1)
	
	(v1) edge node {0} (w1)
	(v2) edge node {0} (w2)
	(vn) edge node {0} (wn)
	
	(w1) edge node [below right] {0} (v2)
	(w2) edge node [below right] {0} (v3)
	(wn) edge node [below right] {0} (vn1)
	
	(vn1) edge node {T} (t)
	;
	\end{tikzpicture}
	\caption{A sketch of the graph $G$ created by the reduction}
	\label{fig:reduction}
\end{figure}

It is not hard to show that for each $x_i$ the number of edges we create in $G$ is linear in $\log x_i$ the size of $G$ is polynomial in the size of the input. Formally, we prove the following claim:

\begin{claim} \label{clm:graphpoly}
The size of $G$ is polynomial in the size of the input.
\end{claim}
\begin{proof}
	We will show that the number of edges between $v_i$ and $v_{i+1}$ is linear in $\log x_i$.
	We know that for $j>2$ the $j$th edge in the sequence has cost at most $\frac{1}{2b}
	\times 2^{j-2}$, and that the sum of the costs of the edges is $x_i$. Let
	$e_i$ be the number of edges between $v_i$ and $v_{i-1}$. Ignoring the first
	edge in the sequence, it is sufficient for $e_i$ to be large enough such that $\sum_{j=1}^{e_i} \frac{1}{2b} \times 2^{j-1} \ge x_i$. Note that,
	\begin{align*}
	\sum_{j=1}^{e_i} \frac{1}{2b} \times 2^{j-1}  = \frac{1}{2b} \sum_{j=0}^{e_i-1} 2^j  = \frac{1}{2b} \p{\frac{2^{e_i} - 1}{2 - 1}} \geq  \frac{1}{6} \p{2^{e_i} - 1}
		\end{align*}
 where the last transition is due to our choice of $b=2+ \lambda \leq 3$. Thus, we have that $	2^{e_i} \ge 6x_i + 1$ and hence the number of required edges between $v_i$ and $v_{i+1}$ is at most $log\p{6x_i + 1}$ and the size of $G$ is polynomial in
 the size of the input.
\end{proof}

In the next two claims we prove that the subset sum instance has a solution $\Longleftrightarrow$ there is
some path in $G$ that the doubly sophisticated agent will traverse for the
given reward.

\begin{claim} \label{clm:redforward}
	The subset sum instance has a solution $\Longrightarrow$ there is some path in
	$G$ that the doubly sophisticated agent will traverse for the given reward.
\end{claim}
\begin{proof}
	Let $I \subseteq \{1, \dots, n\}$ be the solution to the subset sum instance,
	i.e. $\sum_{i \in I} x_i = T$. If there are multiple solutions, then let $I$
	be the one which has the largest minimum index (with ties broken by the
	second-smallest index, and so on). Then, consider a path through $G$ such that
	for $i \in I$, the agent takes the sequence of edges from $v_i$ to $v_{i+1}$
	for a cost of $x_i$, and for $j \notin I$, the agent takes the edges $(v_i,
	w_i)$ and $(w_i, v_{i+1})$. Then, when the agent gets to $v_{n+1}$, it will
	have incurred a total cost of $T$, so when it takes the last edge from
	$v_{n+1}$ to $t$, the final incurred cost will be $2T$. We will show that if
	the agent follows this path, at every vertex that the agent passes through,
	the perceived cost will always be smaller than $R$, meaning this is a valid
	path. Finally, we will show that the agent cannot take any other path.
	
	Call the path described above $P$. We proceed by induction, proving the claim
	that if the agent reaches $v_i$ along $P$, then the agent will reach $t$ along
	$P$.
	
	\noindent \textbf{Base case:} $i=n+1$. If the agent reaches $v_{n+1}$ along $P$, it will
	have incurred a total cost of $T$. The only outgoing edge from $v_{n+1}$ goes
	to $t$ for a cost of $T$, so the perceived cost of taking that edge is $b
	\cdot T$ which is smaller than the perceived reward of $ (b-\lambda) T + \lambda - \eps + \lambda T = b\cdot T + \lambda - \eps$. 
	 Thus, if the agent reaches $v_{n+1}$ along
	$P$, it will reach $t$ along $P$. \\
	\textbf{Inductive hypothesis:} If the agent reaches $v_i$ along $P$, it will
	continue along $P$ to $t$ without abandoning. \\
	\textbf{Inductive step:} Assume that the agent reaches $v_{i-1}$ along $P$.
	Then, in order to prove the claim, we must show that it proceeds from
	$v_{i-1}$ to $v_i$ along $P$, at which point we can use induction to prove
	that from $v_i$, it continues to $t$ along $P$. Since the agent has reached
	$v_{i-1}$ along $P$, we know that at this point, it has incurred a cost of
	$K_{i-1} = \sum_{j \in I, j < i-1} x_j$. We consider two cases: \\
	\textbf{Case 1:} $i-1 \in I$. In this case, we must show that the agent takes
	the sequence of edges from $v_{i-1}$ to $v_i$. Assume towards contradiction
	that the agent takes the $v_{i-1} \to w_{i-1} \to v_i$ path instead. In order
	for $w_{i-1}$ not to be abandoned for this incoming cost, there must be some
	path that reaches $t$ from $w_{i-1}$ with a sunk cost of $K_{i-1}$ without
	abandoning. In order for the agent to take the edge from
	$v_{n+1}$ to $t$, the total cost of the path from $s$ to $v_{n+1}$ should be at least $T$. 

Thus, any valid path from $w_{i-1}$ to $v_{n+1}$ must have a total
	cost of at least $T - K_{i-1}$. However, if this path does have a total cost
	of exactly $T - K_{i-1}$, then following it would yield a valid solution to the subset
	sum instance. Moreover, this solution $I'$ would have a larger minimum index
	than our original solution $I$, as $x_{i-1} \in I$ and $x_{i-1} \notin I'$, and
	for all $j < i-1$, either $j$ is in both $I$ and $I'$ or it is in neither.
	Thus, by contradiction, the path from $w_{i-1}$ to $v_{n+1}$ must have a total cost
	strictly greater than $T - K_{i-1}$. As all edge costs are integers this implies that the path should have a cost of at least $T - K_{i-1} + 1$. As a result the perceived cost of completing this path by going from $v_{n+1}$ to $t$ would is at least $2T - K_{i-1} + 1$. On the other hand the perceived cost of going from $v_{i-1}$ to $v_i$ is 
					\begin{align*}
			b \cdot \frac{1}{2b}  + 2T - K_{i-1} - \frac{1}{2b} 
			&< \frac{1}{2} + 2T - K_{i-1}  
			< 2T - K_{i-1} + 1 
		\end{align*}

   This is because by induction,
  if the agent begins on the path from $v_{i-1}$ to $v_i$, it will follow $P$ from
  $v_i$ to $t$, meaning the total remaining cost is $2T - K_{i-1} - 1/(2b)$.
  Furthermore, the agent will not abandon at $v_{i-1}$, as 
  the perceived cost is
  less than the reward for  $\lambda > 1/2$ and $\eps \le \frac{1}{2b}$ : $b \cdot \frac{1}{2b}  + 2T - K_{i-1} - \frac{1}{2b}  < \frac{1}{2}+2T  - \frac{1}{2b} < 2T+\lambda -\eps$.
	
	For any vertex between $v_i$ and $v_{i+1}$ along the sequence of edges, we
	must show that the perceived cost is no more than the perceived reward. Let $y$ be the
	cost of the edge leading out of some intermediate vertex along this sequence.
	Since the cost of each edge increases by a factor of two, the incurred cost so
	far along the sequence of edges is also $y$. This also means that after
	following this edge of cost $y$, the remaining cost is at most $2T - 2y$
	(because the agent has already incurred a cost of at least $y$, and the next
	edge also has a cost of $y$). Thus, the perceived cost is no more than $by + (2T - 2y) = \lambda y +2T$ while the perceived reward is at least $R+ \lambda y = 2T + \lambda -\eps + \lambda y$ meaning that the agent  does not abandon.

	\noindent
	\textbf{Case 2:} $i-1 \notin I$. As before, we know that any path from
	$v_{i-1}$ to $t$ must have a total cost of at least $2T - K_{i-1}$. 
	The perceived cost of going to $w_{i-1}$ is $2T - K_{i-1} $. The perceived cost of beginning the sequence of edges from $v_{i-1}$
	to $v_i$ is at least  	$b \cdot \frac{1}{2b}  + 2T -
	(K_{i-1} +\frac{1}{2b}) = (b-1) \frac{1}{2b} + 2T - K_{i-1}$ 
	(assuming $x_{i-1} \ge 1$). Therefore, the perceived cost of
	following $P$ is smaller than the perceived cost of following the
	sequence of edges from $v_{i-1}$ to $v_i$, so the agent chooses to go to
	$w_{i-1}$. Furthermore, the agent will not abandon at either $v_{i-1}$ or
	$w_{i-1}$ because at both vertices, the perceived cost is $2T - K_{i-1}
	 < R $.
	
	Since the claim holds in both cases, the induction holds. Thus, if the agent
	follows $P$ to $v_1$, it must continue to follow $P$ until it reaches $t$
	without abandoning. However, the only way to reach $v_1$ from $s$ is along
	$P$. Thus, all that remains to be shown is that the agent will not abandon at
	$s$. However, this must be the case, as the total cost of the path is $2T$ and
	the $(s,v_1)$ edge has cost 0, so the perceived cost at $s$ is $2T < R$.
\end{proof}

\begin{claim} \label{clm:redreverse}
	The subset sum instance has a solution $\Longleftarrow$ there is some path in
	$G$ that the doubly sophisticated agent will traverse for the given reward.
\end{claim}
\begin{proof}
	Let $P$ be the path that the agent traverses. Since $P$ must pass through
	$v_{n+1}$, let $K_{n+1}$ be the sunk cost when the agent reaches $v_{n+1}$.
	The total cost of the path is $K_{n+1} + T$, so the perceived cost for the
	agent at $s$ is $K_{n+1} + T$. In order for the agent to be willing to begin
	along $P$, the perceived cost must be no more than the reward, so $K_{n+1} + T \le 2T + \lambda - \epsilon$. Because $K_{n+1}$ must be an integer this implies that $K_{n+1} \le T$.

	When the agent reaches $v_{n+1}$, in order for it to be willing to take the
	$(v_{n+1},t)$ edge the perceived cost at that point must also be no more than
	the reward, so

	\begin{align*}
b \cdot T  &\le 2T + \lambda - \epsilon + \lambda K_{n+1} 
\end{align*}
By plugging in $b=2+\lambda$ and rearranging we get that $T - 1 + \epsilon/\lambda \le K_{n+1}$ and because $K_{n+1}$ must be an integer we conclude that $T \le K_{n+1}$.

	Since we have both $K_{n+1} \le T$ and $K_{n+1} \ge T$, $K_{n+1} = T$. Then,
	we can construct a solution to the subset sum instance by letting $I = \{i ~ |
	~ P \text{ does not include $w_i$}\}$. $\sum_{i \in I} x_i = K_{n+1} = T$
	because $i \in I$ iff $P$ follows the sequence of edges of total cost $x_i$
	from $v_i$ to $v_{i+1}$, and all other edges in $P$ have cost 0. Thus, $I$ is
	a solution to the subset sum instance.
\end{proof}

\section{Doubly Naive Agents} \label{SEC:DOUBLY-NAIVE}
Consider an agent which has both present bias and sunk cost bias, and is naive
about both. We know that a naive present-biased agent might abandon partway
through a task. This is still the case for a doubly-naive agent, but in contrast
to a naive present-biased agent (whose cost is bounded by $R$), because the
perceived reward keeps increasing the doubly-naive agent can actually incur an
arbitrary amount of cost along the way. We first provide a bound on the cost a
doubly naive agent incurs and then show this upper bound is almost tight.
Note that since the payoff of an optimal agent is at most $R$ the claim
also establishes an asymptotic exponential additive gap between the payoff of a
doubly naive agent and an optimal agent.

\begin{claim}
	  A $(b,\lambda)$-doubly naive agent with $b>1$ and $\lambda >0$ traversing any graph $G$ on $n$ nodes incurs a cost of at most
	$O\p{R\p{1 + \frac{\lambda}{b}}^n}$.
\end{claim}
\begin{proof}
	Consider the path that the doubly naive agent takes through 
	$G$. Let $R_i$ denote the perceived reward when the agent is at node $i$ and
	let $z_i$ denote the cost of the edge that the agent takes leaving $i$. In
	order for the agent to continue at node $i$, it must be the case that the
	perceived cost is no more than $R_i$, so $bz_i \le R_i$. Since $R_{i+1} = R_i
	+ \lambda z_i$, it must be that $R_{i+1} \le R_i (1 + \lambda/b)$. Thus, if
  $R_n$ is the perceived reward when the agent has reached the target, $R_n \le
  R_0(1+\lambda/b)^n$ where $R_0 = R$. Since $R_n - R_0 = \lambda \sum_{i=1}^n
  z_n$, we have $\sum_{i=1}^n z_n \le R((1 + \lambda/b)^n - 1)/\lambda =
  O(R(1+\lambda/b)^n)$.
\end{proof}

Next, we show that the above bound is nearly tight.

\begin{claim} \label{clm:dn-lower-bound}
	There exists a fan type graph on $n$ nodes (Figure \ref{fig:doubly_naive_exp}) in which
  a doubly naive agent with $b>1$ and $\lambda >0$ traversing $G$ incurs cost
  $\Theta\p{R\p{\frac{b(b+\lambda)}{b^2 + \lambda}}^n}$.
\end{claim}
\begin{proof}
  The full proof\full{is deferred to Appendix \ref{app:doubly-naive}.}\camera{can be found in the full version.}
Consider the instance in Figure
\ref{fig:doubly_naive_exp} with the following definitions:
\begin{align*}
  x_i = y_0 \frac{b(b-1)}{b^2 + \lambda} \p{\frac{b(b+\lambda)}{b^2 +
  \lambda}}^{i-1}; &&
  y_i = y_0 \p{\frac{b(b+\lambda)}{b^2 + \lambda}}^i; &&
  R = b y_0
\end{align*}
Note that the costs are increasing exponentially with base
$b(b+\lambda)/(b^2+\lambda)$. Because the perceived reward is increasing as the
agent traverses the graph, it will be willing to continue along the outer edge
of the fan, incurring total cost exponential in the size of the graph.
\end{proof}
	
\begin{figure}[ht]
  \centering
  \begin{tikzpicture}[->,shorten >=1pt,auto,node distance=2cm, thin, scale=.5]
    \node (0) [circle, draw, minimum size=.8cm] at (0,0) {$s$};
    \node (1) [circle, draw, minimum size=.8cm] at (1,3) {$v_1$};
    \node (2) [circle, draw, minimum size=.8cm] at (5,5) {$v_2$};
    \node (3) [circle, draw, minimum size=.8cm] at (9,3) {$v_3$};
    \node (4) [circle, draw, minimum size=.8cm] at (10,0) {$v_4$};
    \node (5) [circle, draw, minimum size=.8cm] at (5,0) {$t$};
    
    \path
    (0) edge node {$x_1$} (1)
    (1) edge node {$x_2$} (2)
    (2) edge node {$x_3$} (3)
    (3) edge node {$x_4$} (4)

    (0) edge node {$y_0$} (5)
    (1) edge node {$y_1$} (5)
    (2) edge node {$y_2$} (5)
    (3) edge node {$y_3$} (5)
    (4) edge node {$y_4$} (5)
    ;
  \end{tikzpicture}
  \caption{Graph for which doubly naive agent incurs exponential cost}
  \label{fig:doubly_naive_exp}
\end{figure}
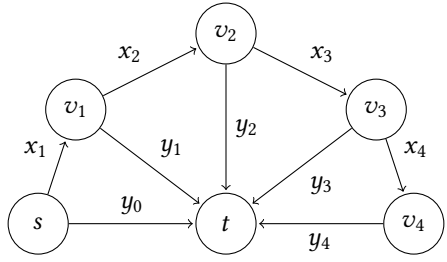

It is not always the case that a doubly naive agent does worse than a naive
present-biased agent. As we will see next, sometimes sunk cost bias may
actually help the agent reach the goal and achieve a positive payoff. Next, we bound the possible gain due to sunk cost bias and prove the following claim:
	
\begin{claim} \label{clm:naive-vs-doubly}
	The payoff of a doubly naive agent can be greater by $R(1-\frac{1}{b}))$ than the payoff of a
  naive present-biased agent. The bound is tight.
\end{claim}
\begin{proof}
	Recall that we would like to prove that the payoff of a doubly naive agent can be greater by $R(1-\frac{1}{b}))$ than the payoff of a naive present-biased agent.
	Observe that a naive present-biased agent and a doubly naive agent will take
	the exact same path till the point that the naive present-biased agent
	abandons. Denote the node at which this happens by $v_i$ and the node that the
	doubly naive agent will continue to by $v_{i+1}$. Since the naive
	present-biased agent abandons we have that $b \cdot c(v_i,v_{i+1})+C_o(v_{i+1})
	> R$ (where $C_o(v_i)$ is the cost of the optimal agent for reaching $t$ from node $v_i$). This in particular implies that $c(v_i,v_{i+1})+C_o(v_{i+1}) >
	\frac{R}{b}$. As $c(v_i,v_{i+1})+C_o(v_{i+1}) $ is a lower bound on the cost of
	the doubly naive agent for getting from $v_{i+1}$ to $t$. We get that the gap
	between the two agents is at most $R-\frac{R}{b}$, since from $v_i$ the naive
	present-biased agent gets a payoff of 0, and the doubly naive agent gets a
	payoff of at most $R - c(v_i, v_{i+1}) - C_o(v_{i+1}) < R - \frac{R}{b}$.
	
	The instance in Figure \ref{fig:doubly-naive} illustrates this bound is tight. Observe that a naive present-biased agent will stop traversing the graph at node $v$ as $b \cdot c < R$. A doubly naive agent will get to $t$ for a total payoff of $R-\frac{R+ \lambda \cdot \eps}{b}-\eps = R(1-\frac{1}{b}) - (1+\lambda/b)\eps $. 
	
	\begin{figure}[ht] 
		\centering 
		\begin{tikzpicture}[->,shorten >=1pt,auto,node distance=2cm, thin] 
		\node (s) [circle, draw, minimum size=1cm] at (0,0) {$s$}; 
		\node (v) [circle, draw, minimum size=1cm] at (2,0) {$v$}; 
		\node (t) [circle, draw, minimum size=1cm] at (4,0) {$t$};
		
		\path (s) edge node {$\eps$} (v) 
		(v) edge node {$\frac{R+ \lambda \cdot \eps}{b}$} (t) ; 
		\end{tikzpicture} 
		\caption{For any $\lambda>0$ and $R>0$ a doubly naive agent will traverse the graph but a naive agent will abandon at $v$.}
		\label{fig:doubly-naive}
	\end{figure}
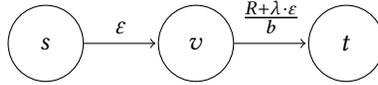
\end{proof}

\section{Singly Sophisticated Agents} \label{SEC:SINGLY-SOPH}
A singly sophisticated agent is an agent who is sophisticated about its present
bias but is naive about its sunk cost bias. In contrast to sophisticated
present-biased agents, singly sophisticated agents can't plan their whole path
ahead of time, as their beliefs about the reward are time-inconsistent. Such an
agent, for example, can abandon a task, as demonstrated in Figure
\ref{fig:sing-abandons}. In this example a singly sophisticated agent standing
at $s$ will plan to follow the path $s \rightarrow u \rightarrow v \rightarrow
t$ as it believes that since the reward is only $11$, its future self at $v$
will go straight to $t$ instead of taking the $v \to w \to t$ path. However,
because of sunk cost bias, once it reaches $u$ its perceived reward is increased
to $12$. As a result, it now believes that its future self at $v$ will take the
path $v \rightarrow w \rightarrow t$. When standing at $u$, the perceived cost of
this path is $13>12$, and hence the agent abandons.

	\begin{figure}[ht] \centering 
	\begin{tikzpicture}[->,shorten >=1pt,auto,node distance=2cm,
	thin, scale=.8] 
	\node (s) [circle, draw, minimum size=.8cm] at (-2,0) {$s$}; 
	\node (u) [circle, draw, minimum size=.8cm] at (0,0) {$u$}; 
	\node (v) [circle, draw, minimum size=.8cm] at (2,0) {$v$}; 
	\node (w) [circle, draw, minimum size=.8cm] at (3,-1.5) {$w$}; 
	\node (t) [circle, draw, minimum size=.8cm] at (4,0) 	{$t$};
	
	\path (s) edge node {$2$} (u) 
	(u) edge node {$4$} (v) 
	(v) edge node {0} (w) 
	(v) edge node {3} (t) 
	(w) edge node {6} (t) ; 
	\end{tikzpicture}
	\caption{For $b=2$ and $R=11$, a singly sophisticated agent with $\lambda =1/2$ would abandon. } \label{fig:sing-abandons}
\end{figure}
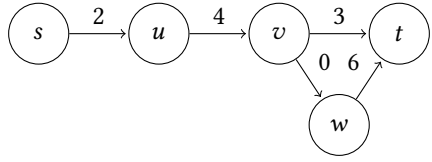

Next, we bound the loss of a singly sophisticated agent in comparison to an optimal agent. We will later show that this bound is close to tight by putting together many instances similar to the one in Figure \ref{fig:sing-abandons} to amplify the sunk cost of a singly sophisticated agent.

\begin{claim} \label{clm:singly-opt}
  The additive gap between the payoffs of an optimal agent and a singly
  sophisticated agent is at most $\frac{(1+\lambda)^{k} - 1}{\lambda} R+R$.
\end{claim}
\begin{proof}
Recall that a singly sophisticated agent can plan to take one path and then
change its plan. Denote the nodes in which the agent decides to change the path
it takes by $v_1,\ldots,v_k$ and let $s=v_0$. Denote the cost accumulated
between $v_{i-1}$ and $v_i$ by $c_{i-1}$. Note that between every two adjacent
change points $v_{i-1}$ and $v_i$ the agent behaves the same as a present-biased
sophisticated agent with no sunk cost bias, and hence its cost is less than the
perceived reward at node $v_{i}$, which we denote by $R_{i}$. Thus we have
that
$c_i \leq R_{i}$ and $ R_i = R_{i-1}+\lambda c_{i-1}$.
Putting this together we get that $R_i \leq (1+\lambda) R_{i-1}$ which implies
that $R_i \leq (1+\lambda)^{i} \cdot R$. Hence,
\begin{align*}
\sum_{j=1}^k c_j \leq \sum_{j=1}^k (1+\lambda)^{j-1} \cdot R = \frac{(1+\lambda)^{k} - 1}{\lambda}  \cdot R
\end{align*}
This concludes the claim as the payoff of an optimal agent is at most $R$.
\end{proof}
In the next claim we show that this gap is essentially tight:

\begin{claim}
	The additive gap between a singly sophisticated agent with $b>2$ and $\lambda >0$ and an optimal agent can be as high as $\frac{(1+\alpha \cdot \lambda)^{n} - 1}{\lambda} R + R(1-\alpha - \frac{1}{b}) $ where $\alpha = \min \{\frac{1}{2b \lambda}, \frac{b-1}{b^2+2\lambda}  \}$.
\end{claim}
\begin{proof}
		\begin{figure}[ht] \centering 
		\begin{tikzpicture}[->,shorten >=1pt,auto,node distance=2cm,
		thin, scale=0.75]
		\node (s) [circle, draw, minimum size=.8cm] at (-2,0) {$s$}; 
		\node (v1) [circle, draw, minimum size=.8cm] at (0,0) {$v_1$}; 
		\node (u1) [circle, draw, minimum size=.8cm] at (0,-2) {$u_1$}; 
		\node (t1) [circle, draw, minimum size=.8cm] at (0,-4) {$t_1$}; 
		\node (w1) [circle, draw, minimum size=.8cm] at (2,-3) {$w_1$}; 
		\node (v2) [circle, draw, minimum size=.8cm] at (4,0) {$v_2$}; 
		\node (u2) [circle, draw, minimum size=.8cm] at (4,-2) {$u_2$}; 
		\node (t2) [circle, draw, minimum size=.8cm] at (4,-4) {$t_2$}; 
		\node (w2) [circle, draw, minimum size=.8cm] at (6,-3) {$w_2$}; 
		\node (t) [circle, draw, minimum size=.8cm] at (2,-6) 	{$t$};
		\node (vn) [circle, draw, minimum size=.8cm] at (8,0) {$v_n$}; 
		\node (un) [circle, draw, minimum size=.8cm] at (8,-2) {$u_n$}; 
		\node (tn) [circle, draw, minimum size=.8cm] at (8,-4) {$t_n$}; 
		\node (wn) [circle, draw, minimum size=.8cm] at (10,-3) {$w_n$}; 
		
		\path (s) edge node {$x_1$} (v1) 
		(v1) edge node {$x_2$} (v2) 
		(v1) edge node {$y_1$} (u1) 
		(u1) edge node {$z_1$} (t1) 
		(w1) edge node {$b \cdot z_1$} (t1)
		(u1) edge node {$0$} (w1) 
		(v2) edge node {$y_2$} (u2) 
		(u2) edge node {$z_2$} (t2) 
		(w2) edge node {$b \cdot z_2$} (t2)
		(u2) edge node {$0$} (w2) 
		(t1) edge node {$0$} (t) 
		(t2) edge node {$0$} (t) 
		(tn) edge node {$0$} (t) 
		(v2) edge node {$\ldots$} (vn) 
		(vn) edge node {$y_2$} (un) 
		(un) edge node {$z_n$} (tn) 
		(un) edge node {$0$} (wn) 
		(wn) edge node {$b \cdot z_n$} (tn)
		; 
		\end{tikzpicture}
		\caption{A graph in which a singly sophisticated agent incurs an exponential cost } \label{fig:sing-vs-opt}
	\end{figure}
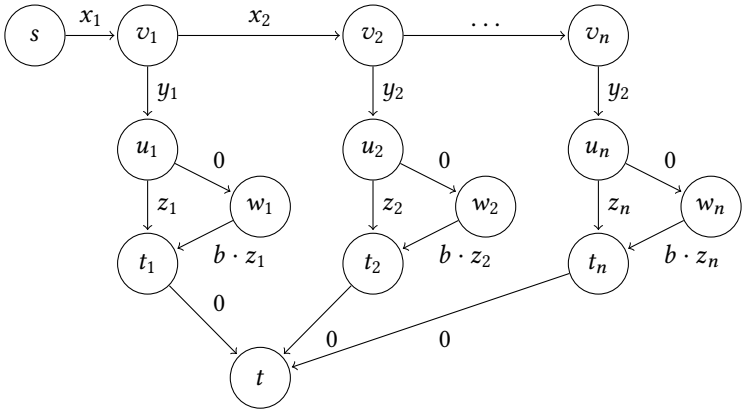
	Consider the instance in Figure \ref{fig:sing-vs-opt}. For simplicity, we define the different costs as a function of the perceived rewards, where as before $R_i$ is the perceived reward at node $v_i$: 
	\begin{align*}
	x_{i} = \alpha \cdot R_{i-1} ; && z_i = \frac{R_{i-1}}{b^2}+\eps; && y_i = \frac{R_i}{b}-\frac{R_{i-1}}{b^2}; &&  R_0 = R
\end{align*}

	We first show that a singly sophisticated agent will take the path $v_1 \rightarrow v_2 \rightarrow \ldots \rightarrow v_n$ and then abandon at $v_n$. 
In particular at every node $v_{i-1}$ (for this purpose $s=v_0$) the agent plans to follow the path $v_{i-1} \rightarrow v_{i} \rightarrow u_{i} \rightarrow t_{i} \rightarrow t$. Once it reaches $v_i$ since the perceived reward has increased, the path  $v_{i} \rightarrow u_{i} \rightarrow t_{i} \rightarrow t$ is no longer an option and the only path it can take to get to $t_i$ is $v_{i} \rightarrow u_{i} \rightarrow w_i \rightarrow t_{i} \rightarrow t$. As the latter path has a perceived cost greater than the perceived reward, the agent plans to follow the path $v_{i} \rightarrow v_{i+1} \rightarrow u_{i+1} \rightarrow t_{i+1} \rightarrow t$ instead. 

\begin{lemma} 
	A $(b,\lambda)$-singly sophisticated agent traversing the graph in Figure \ref{fig:sing-vs-opt} will follow the path $v_1 \rightarrow v_2 \rightarrow \ldots \rightarrow v_n$ and then abandon at $v_n$.
\end{lemma}
\begin{proof}
We show that for every $i$ once the agent reaches $v_i$ because the perceived reward has increased the path  $v_{i} \rightarrow u_{i} \rightarrow t_{i} \rightarrow t$ is no longer an option and the only path it can take to get to $t_i$ is $v_{i} \rightarrow u_{i} \rightarrow w_i \rightarrow t_{i} \rightarrow t$ which for an agent at $v_i$ has a perceived cost greater than the perceived reward. Note that when the agent is standing at $v_{i-1}$ its perceived reward is $R_{i-1}$ and that $R_i = R_{i-1} + \lambda \alpha R_{i-1} = (1+ \alpha\lambda)R_{i-1}$ :	
	\begin{itemize}
		\item For a reward of $R_{i-1}$ an agent at $u_{i}$ will continue to $t_{i}$ - first we observe that the perceived cost of continuing straight to $t_i$ is less than $R_{i-1}$: $b \cdot z_i = \frac{R_{i-1}}{b}+\eps < R_{i-1}$. Second, note that the agent at $v_{i-1}$ believes that the path $u_i \ra w_i \ra t_i$ is not an option since $b^2\cdot z_i = R_{i-1}+\eps \cdot b^2 > R_{i-1}$.
		\item For a reward of $R_{i-1}$ an agent at $v_{i}$ will continue to $u_{i}$ - the agent believes that if it will continue to $u_i$ it will then continue to $t_i$ and then to $t$. The perceived cost of this path is 
		\begin{align*}
			b\cdot y_i +z_i &=  b \cdot (\frac{R_i}{b}-\frac{R_{i-1}}{b^2}) + \frac{R_{i-1}}{b^2}+\eps \\
			&=  R_i-\frac{R_{i-1}}{b} + \frac{R_{i-1}}{b^2}+\eps \\
			&= R_{i-1} + \alpha\lambda R_{i-1}-\frac{R_{i-1}}{b} + \frac{R_{i-1}}{b^2}+\eps \\
			& \leq R_{i-1} + \lambda \cdot R_{i-1} \cdot \frac{1}{2b \lambda} -\frac{R_{i-1}}{b} + \frac{R_{i-1}}{b^2}+\eps \\
			&= R_{i-1} \cdot  \frac{2b^2-b+2}{2b^2} + \eps
		\end{align*}
		For $b>2$ and an appropriate value of $\eps$ the above perceived cost is less than $R_{i-1}$. By construction it is easy to see that the perceived cost of any other path (i.e., continuing from $v_i$ to $v_{i+1} $)  is greater.
		\item For a reward of $R_{i-1}$ an agent at $v_{i-1}$ will continue to $v_{i}$:
		\begin{itemize}
			\item The perceived cost of the path $v_{i-1} \rightarrow v_{i} \rightarrow u_{i} \rightarrow t_{i} \rightarrow t$ is less than $R_{i-1}$:
			\begin{align*}
				b \cdot x_i + y_i + z_i  &= b \cdot R_{i-1} \cdot x_i + \frac{R_{i}}{b} +\eps \\ 
				&= b \cdot x_i + \frac{R_{i-1} + \lambda x_i}{b} +\eps \\ 
				&= \frac{R_{i-1}}{b}  + x_i \frac{b^2+\lambda}{b} + \eps \\
				&= \frac{R_{i-1}}{b}  +\alpha R_{i-1} \cdot \frac{b^2+\lambda}{b} + \eps \\
				& \leq  \frac{R_{i-1}}{b} +  \frac{b-1}{b^2+2\lambda} \cdot \frac{b^2+\lambda}{b} \cdot R_{i-1} + \eps  
			\end{align*}
			The last expression is less than $R_{i-1}$ for an appropriate choice of $\eps$.
			\item For a reward of $R_{i-1}$ an agent at $u_{i-1}$ will continue to $w_{i-1}$. For this it suffices to show that $b^2 z_{i-1} = R_{i-2} + \eps  < R_{i-1}$.
			\item For a reward of $R_{i-1}$ and an agent at $v_{i-1}$ the perceived cost of the path
			$v_{i-1} \to u_{i-1} \to w_{i-1} \to t_{i-1} \to t$ is greater than
			$R_{i-1}$:
			\begin{align*}
				b \cdot y_{i-1} + b \cdot z_{i-1} = b \cdot \frac{R_i}{b} + b \cdot \eps = R_i + b\eps> R_{i-1}.
			\end{align*}
		\end{itemize}
	\end{itemize}
\end{proof}

To compute the total cost the agent incurred, recall that $R_i = R_{i-1} + \lambda \alpha R_{i-1}$ this implies that $(1+\alpha \lambda)^{i} \cdot R$. Note that 
\begin{align*}
\sum_{i=1}^n x_i  = \frac{R_n-R}{\lambda} = \frac{(1+\alpha \lambda)^{n} -1}{ \lambda} \cdot R
\end{align*}
Lastly, observe that an optimal agent will take the path $s \ra v_1 \ra u_1 \ra t_1 \ra t$ for a payoff of $R(1-\alpha-\frac{1}{b})$.
\end{proof}

To understand the role naivete regarding sunk cost plays in agents that are sophisticated about their present bias we now compare between the payoff of singly sophisticated agents and doubly sophisticated agents. By Proposition \ref{prop:double_cost_ratio} we have that the payoff of a doubly sophisticated agent is at most an additive amount of $(b-1)C_o$ from the payoff of an optimal agent. Hence by Claim \ref{clm:singly-opt} we have that a singly sophisticated  can do worse than a doubly sophisticated agent by an exponential additive factor. However, in some cases being naive about its sunk cost can actually help the agent avoid taking a more costly path or reach the target. As the payoff of a doubly sophisticated agent is at most an additive factor of $(b-1)C_o$ of from the payoff of an optimal agent, even in cases in which a singly sophisticated agent surpasses a doubly sophisticated one, the payoff of the singly sophisticated agent is greater by at most $(b-1)C_o$. The example in Figure \ref{fig:sing-better} in which a singly sophisticated agent will behave just as a present-biased sophisticated agent establishes this is tight. This proves the following claim:

\begin{claim}
	The payoff of a singly sophisticated agent can be better than the payoff of a doubly sophisticated agent by an additive amount of $(b-1) C_{o}$.
\end{claim}

\section{Conclusion}
We have studied the interaction between two behavioral biases that
both play an important role in planning contexts: 
present bias and sunk-cost bias.
We find that in conjunction, they give rise to natural behavioral phenomena 
that cannot be seen with either in isolation. Through these biases, we also
gain new insights about subtleties in the behavior of naive and
sophisticated agents. We show that sophistication about these two
biases makes path-planning computationally hard, though we are still
able to provide performance bounds for agents with different forms of
sophistication.

This work leads to several open questions. While we showed that path-planning
for doubly sophisticated agents is NP-hard, we do not know whether or not the
problem is in NP. 
Moreover, in the case where the reward $R$ exceeds $b$ times the optimal 
cost --- which implies that a feasible
path for the doubly sophisticated agent is guaranteed to exist --- 
is it possible to find such a path efficiently? 
We also showed that there exist instances in
which singly sophisticated agents incur exponentially large cost; 
is there a structural characterization (e.g. a graph minor result) 
for the graphs on which singly sophisticated agents incur exponential cost,
in the style of \cite{ko-ec14-full}?
More broadly, the rich interplay between these two biases
demonstrates how considering multiple behavioral biases together 
can yield a wider set of natural phenomena.
With this in mind, we believe there are many further opportunities to
enhance theoretical models of behavior through the analysis of
agents with multiple biases.

\bibliographystyle{ACM-Reference-Format}
\bibliography{header,refs}

\newpage
\begin{appendix}
	\section{Recursive Algorithm for Doubly Sophisticated Path Computation} \label{app:rec-alg}
Let $k$ be an upper bound on the number of possible different values of
  the costs for reaching a node.  \aref{intdsrec} recursively computes the path
that a doubly sophisticated agent will take. Note that as the number of
subproblems the algorithm need to solve is at most $k \cdot n$, the algorithm
runs in time polynomial in $n$ and $k$. Hence, for example, if the number of
paths connecting $s$ and $t$ is polynomial then the path computation problem can
be solved in polynomial time. 
\alg{\textsc{RecursiveIntegerDoublySophisticated}$(G, R, b, \lambda)$}{intdsrec}
\State Initialize $choices$ and $costs$ to empty hashmaps
\Procedure{ComputePathAndCosts}{$u, i$}
\If{$u == t$}
\State $costs[u, i] = 0$
\State \Return
\EndIf
\For{$v \in N(u)$}
\If{$choices[v, i + c(u,v)]$ is empty}
\State \Call{ComputePathAndCosts}{$v, i + c(u,v)$}
\EndIf
\EndFor
\State $v \gets \argmin_{v' \in N(u)} b \cdot c(u,v') + costs[v', i+c(u,v')]$
\State $perceived \gets b \cdot c(u,v) + costs[v, i+c(u,v)]$
\If {$perceived > R + \lambda \cdot i$}
\State $choices[u, i] \gets null$
\State $costs[u, i] \gets \infty$
\Else
\State $choices[u, i] \gets v$
\State $costs[u, i] \gets c(u,v) + costs[v, i+c(u,v)]$
\EndIf
\EndProcedure
\State \Call{ComputePathAndCosts}{$s, 0$}
\ealg
  \section{Supplementary Material for Section \ref{SEC:DOUBLY-NAIVE}}
  \label{app:doubly-naive}
	
	\prevproof{Claim}{clm:dn-lower-bound}
	{
Consider Figure \ref{fig:doubly_naive_exp} with costs given by
\begin{align*}
  x_i &= y_0 \frac{b(b-1)}{b^2 + \lambda} \p{\frac{b(b+\lambda)}{b^2 +
  \lambda}}^{i-1} \\
  y_i &= y_0 \p{\frac{b(b+\lambda)}{b^2 + \lambda}}^i \\
  R &= b y_0
\end{align*}
Note that
\begin{align*}
  \sum_{j=1}^i x_j &= y_0 \frac{b(b-1)}{b^2 + \lambda} \sum_{j=1}^i
  \p{\frac{b(b+\lambda)}{b^2+\lambda}}^{j-1} \\
  &= y_0 \frac{b(b-1)}{b^2 + \lambda} \p{\frac{\p{\frac{b(b+\lambda)}{b^2 +
  \lambda}}^i - 1}{\frac{b(b+\lambda)}{b^2+\lambda} - 1}} \\
  &= y_0 \frac{b(b-1)}{b^2 + \lambda} \p{\frac{\p{\frac{b(b+\lambda)}{b^2 +
  \lambda}}^i - 1}{\frac{\lambda(b-1)}{b^2+\lambda}}} \\
  &= \frac{by_0}{\lambda} \p{\p{\frac{b(b+\lambda)}{b^2+\lambda}}^i - 1}
  \numberthis \label{eq:sumxj}
\end{align*}

At node $v_i$, the agent has incurred a cost of $\sum_{j=1}^i x_i$. It must
choose between going directly to $t$, for a perceived cost of $b y_i$, or going
to $v_{i+1}$ and then to $t$, for a perceived cost of $b x_{i+1} + y_{i+1}$.
However, by construction,
\begin{align*}
  b x_{i+1} + y_{i+1} &= y_0 \p{\frac{b^2(b-1)}{b^2 + \lambda}
  \p{\frac{b(b+\lambda)}{b^2 + \lambda}}^i + \p{\frac{b(b+\lambda)}{b^2 +
  \lambda}}^{i+1}} \\
  &= y_0 \p{\frac{b^2(b-1)}{b^2 + \lambda}
  \p{\frac{b(b+\lambda)}{b^2 + \lambda}}^i + \frac{b(b+\lambda)}{b^2 +
  \lambda}\p{\frac{b(b+\lambda)}{b^2 + \lambda}}^i} \\
  &= y_0 \p{\frac{b^2(b-1) + b(b+\lambda)}{b^2 + \lambda}
  \p{\frac{b(b+\lambda)}{b^2 + \lambda}}^i} \\
  &= by_i \p{\frac{b^2 - b + b + \lambda}{b^2 + \lambda}} \\
  &= by_i
\end{align*}
Breaking ties by continuing to $v_{i+1}$, the agent will always prefer to
continue along the fan. The reward is always large enough to do so, because the
perceived reward at $v_i$ is 
\begin{align*}
  R + \lambda\sum_{j=1}^i x_j &= b y_0 + b y_0 \p{\p{\frac{b(b+\lambda)}{b^2 +
  \lambda}}^i - 1} \\
  &= b y_0 \p{\frac{b(b+\lambda)}{b^2 + \lambda}}^i \\
  &= by_i \\
  &= b x_{i+1} + y_{i+1}
\end{align*}
Thus, the total cost incurred by the agent is
\begin{align*}
  \sum_{i=1}^n x_i + y_n &= y_0 \b{\frac{b}{\lambda}
  \p{\p{\frac{b(b+\lambda)}{b^2 + \lambda}}^n - 1} + \p{\frac{b(b+\lambda)}{b^2
  + \lambda}}^n} \\
  &= \frac{y_0}{\lambda} \b{(b+\lambda) \p{\frac{b(b+\lambda)}{b^2
  + \lambda}}^n - b} \\
  &= \Theta\p{R\p{\frac{b(b+\lambda)}{b^2 + \lambda}}^n} \\
  &= \Theta\p{R\p{1 + \frac{(b-1)\lambda)}{b^2 + \lambda}}^n}
\end{align*}
Note that for this family of examples, a naive present-biased agent with only
present bias (or with $\lambda = 0$) would go from $s$ to $v_1$ and then immediately
abandon because the perceived cost would be higher than the reward.
}

\end{appendix}

\end{document}